\documentclass[openacc]{rsproca_new}%%%%where rsproca is the template name

\usepackage{tikz}
\usepackage{subfig}
\usetikzlibrary{shapes.geometric,calc}

\newcommand{\be}{\begin{equation}}
\newcommand{\en}{\end{equation}}

\def\d{{\rm d}}
\def\uv{{\boldsymbol u}}

\def\rv{{\boldsymbol r}}
\def\fv{{\boldsymbol f}}

\def\xv{{\boldsymbol x}}
\def\yv{{\boldsymbol y}}
\def\nv{{\boldsymbol n}}

\def\dv{{\boldsymbol d}}

\def\lv{{\boldsymbol l}}

\def\ev{{\boldsymbol e}}
\def\Bv{{\boldsymbol B}}

\def\Cv{{\boldsymbol C}}
\def\Jv{{\boldsymbol J}}

\def\Fv{{\boldsymbol F}}
\def\Av{{\boldsymbol A}}
\def\Tv{{\boldsymbol T}}

  \def\Nv{{\boldsymbol N}}
    
    \def\Gv{{\boldsymbol G}}
 \def\Wv{{\boldsymbol W}}

\def\cl{{\mathcal{L}}}

\def\nvec{\Bv_{f\perp}}
\def\grad{\boldsymbol\grad}

\def\wh{\mathcal{H}}

\def\flh{{\mathcal{H}}}
\def\flw{{\mathcal{L}}}

\def\grad{{\rm grad}\, }

\def\bgamma{{\boldsymbol\gamma}}
\newcommand{\deriv}[2]{\frac{\d #1}{\d #2}}
\newcommand{\pder}[2]{\frac{\partial #1}{\partial #2}}

\newtheorem{theorem}{\bf Theorem}[section]

\begin{document}

%%%% Article title to be placed here
\title{Magnetic winding: what is it and what is it good for?}

\author{%%%% Author details
C. Prior$^{1}$ and D. MacTaggart$^{2}$}

%%%%%%%%% Insert author address here
\address{$^{1}$Department of Mathematical Sciences, Durham University, Durham, DH1 3LE, UK\\
$^{2}$School of Mathematics and Statistics, University of Glasgow, Glasgow G12 8QQ, UK\\
}

%%%% Subject entries to be placed here %%%%
\subject{magnetohydrodynamics, topological fluid dynamics}

%%%% Keyword entries to be placed here %%%%

%%%% Insert corresponding author and its email address}
\corres{D. MacTaggart\\
\email{david.mactaggart@glasgow.ac.uk}}

%%%% Abstract text to be placed here %%%%%%%%%%%%
\begin{abstract}
Magnetic winding is a fundamental topological quantity that underpins magnetic helicity and measures the entanglement of magnetic field lines. Like magnetic helicity, magnetic winding is also an invariant of ideal magnetohydrodynamics. In this article we give a detailed description of what magnetic winding describes, how to calculate it and how to interpret it in relation to helicity. We show how magnetic winding provides a clear topological description of magnetic fields (open or closed) and we give examples to show how magnetic winding and helicity can behave differently, thus revealing different and imporant information about the underlying magnetic field.  
% describe how it can be used in practical solar physics applications. {\color{red}We show different physical properties of magnetic fields, such as axial current and Lorentz forces affect the magnetic helicity and winding in differing ways.} 
%The nonlinear evolution of magnetohydrodynamic (MHD) instabilities is complicated by the development of secondary instabilities that exhibit characteristics of both ideal and resistive MHD. Classical theories of stability are difficult to apply in such scenarios and new approaches are required. For certain domains and parameter regimes, an analysis of magnetic topology can reveal detailed information on both `ideal' and `resistive' aspects of instabilities. In this work, we study the evolution of magnetic helicity and magnetic winding in a kink-unstable magnetic field.  
\end{abstract}
%%%%%%%%%%%%%%%%%%%%%%%%%%%

%%%%%%%%%% Insert the texts which can accomdate on firstpage in the tag "fmtext" %%%%%

\begin{fmtext}
\section{Introduction}
\label{S:1}
The title of this paper pays homage to the now classic article by Finn and Antonsen (FA) \cite{finn1985magnetic} which, together with the seminal work of Berger and Field (BF) \cite{berger1984topological}, introduced \emph{relative magnetic helicity} - an important topological invariant of ideal magnetohydrodynamics (MHD). The ``what is it'' of FA's title describes how helicity can be defined for a magnetic field with non-tangential components on domain boundaries. BF show that by extending the magnetic field so that it becomes closed outside the domain (adding a closure), a relative measure of (gauge-invariant) helicity can be found which compares two different magnetic fields with the same bounday conditions and closure. The general formula for relative helicity, that is most widely used today, is presented in FA. For a (simply connected) domain $\Omega$, the relative magnetic helicity $H_R$ is given by
\be\label{hr}
H_R = \int_{\Omega}(\Av+\Av')\cdot(\Bv-\Bv')\,\d^3x,
\en
where $\Bv$ and $\Bv'$ are divergence-free fields (magnetic fields) with the same boundary conditions on $\partial\Omega$ with
\end{fmtext}
\maketitle
\noindent $\Bv=\nabla\times\Av$ and $\Bv'=\nabla\times\Av'$.

The ``what is it good for'' is described in FA for fusion applications. However, equation (\ref{hr}) has been used heavily in solar physics to understand the topological properties of magnetic fields in the solar atmosphere. A recent review of calculating relative helicity in solar applications can be found in \cite{pariat2020chapter}.

Solar observations cannot determine $H_R$ directly as it is presented in equation (\ref{hr}). This is because three-dimensional (3D) information about the magnetic field is not available throughout the solar atmosphere. This information can only be found at the solar boundary, the photosphere, and so a full 3D magnetic field needs to be constructed based on a particular model. The chosen model is normally a force-free field \cite{wheatland2000optimization,guo2012modeling,wiegelmann2012solar}, which is only strictly justified in the corona and not in the photosphere. That being said, other non-force-free models have also been considered \cite{wiegelmann2017mhs,neukirch2019mhs}.

Another way that the relative helicity can be determined from solar observations is to integrate the rate of change of relative helicity through the photosphere (where the magnetic field components can be observed). As shown in BF, the rate of change of relative helicity through a horizontal plane $P$ (representing the photosphere), in ideal MHD,  can be written as
\be\label{hr_rate}
\frac{\d H_R}{\d t} = 2\int_P[(\Av'\cdot\Bv)u_z-(\Av'\cdot\uv)B_z]\, d^2x,
\en
where $\ev_z$ is normal to $P$, $\uv$ is the velocity of the flow and $\Av'$ satisfies the following chosen properties: $\nabla\times\Av'\cdot\ev_z = B_z$ and $\nabla_{\perp}\cdot(\ev_z\times\Av'\times\ev_z) = 0$ on $P$, where $\nabla_{\perp}$ refers to the horizontal gradient on $P$.

Although equation (\ref{hr_rate}) is not the most general expression of the rate of relative helicity through a plane, it does have particular advantages {that will lead us to uncover the underlying topological structure of relative helicity}. Before stating what these advantages are, let us quickly recap some useful topological insights concerning classical helicity. Moffat \cite{moffatt1969degree} showed that classical helicity,
\be\label{class_gen}
H=\int_{\Omega} \Av\cdot\Bv\,\d V,
\en
where $\Bv\cdot\nv=0$ on $\partial\Omega$, has a topological interpretation in terms of the Gauss linking number. Consider two linked loops (a pair of closed and linked magnetic field lines) given by $\d \xv/\d s =\Bv(\xv)$ and $\d \yv/\d s =\Bv(\yv)$ (the following can be extended to $n$ loops and also ergodic field lines, but two loops will suffice for this demonstration). Then the Gauss linking number \cite{ricca2011gauss} is given by
\be\label{gauss}
Lk(\xv,\yv)= \frac{1}{4\pi}\int_{\xv(s)}\int_{\yv(\sigma)}\frac{\d\xv}{\d s}\cdot\frac{\d\yv}{\d \sigma}\times\frac{\xv-\yv}{|\xv-\yv|^3}\,\d s\,\d\sigma.
\en
If around each loop we identify a solid magnetic torus, as in \cite{ricca2008knots}, then since the magnetic field is tangential to the toroidal boundary, we can consider the domain to be $\mathbb{R}^3$, where the magnetic fields in the tori are `extended by zero' outside of the tori volumes \cite{mactaggart2019multi}. Thus equation (\ref{gauss}) can lead to classical helicity formula
\be\label{class}
H = \frac{1}{4\pi}\int_{\mathbb{R}^3}\int_{\mathbb{R}^3}\Bv(\xv)\cdot\Bv(\yv)\times\frac{\xv-\yv}{|\xv-\yv|^3}\, \d^3x\, \d^3y.
\en
Equation (\ref{class}) is {equivalent to} equation (\ref{class_gen}) {but is written in terms of a} particular gauge, namely $\nabla\cdot\Av=0$. However, the purpose of ``deriving'' equation (\ref{class}) in the way above is to show that helicity has an underlying topological structure (here measured by the Gauss linking number) that is not dependent on the magnetic field strength. The magnetic field can be constructed (at least in a formal way) from the underlying geometric structure of the field lines, whose topology is described by Gauss linkage. Once the loops are turned into genuine magnetic field volumes, the classical helicity $H$ becomes a measure that \emph{combines} magnetic field strength \emph{and} topological information (linking).

Returning to equation (\ref{hr_rate}), although it is not the most general expression of the relative helicity rate {(in the sense that it is based on a particular gauge)}, the choices made for $\Av'$ reveal an underlying topological structure similar to that for classical helicity described above. Berger \cite{berger1986braids} was the first to notice this underlying structure and his ideas were developed further by Prior and Yeates \cite{prior2014helicity}. We refer to this underlying topological structure as \emph{magnetic winding}. 

Just as classical helicity can be thought to be based on the more fundamental property of Gauss linkage, in this work, we will show that relative helicity (through equation \ref{hr_rate}) can be based on the more fundamental property of winding. We show that magnetic winding can provide the underlying topological structure for both open and closed magnetic fields and we will give examples of how this property can be used in applications. {In particular, we will show that the calculation of magnetic winding in solar flux emergence can be used to identify regions of complex field line topology (see also \cite{prior2019interpreting,mactaggart2020helicity}). This result can then be used to aid the prediction of solar eruptions, which depend on non-trivial field line topology.}

The outline of the paper is as follows: {we first will define magnetic winding, which provides the underlying topological descriptions for quantities to be introduced later. We then introduce \emph{winding helicity} and show how this provides a link bewteen winding, relative helicity and relative helicity flux. We then renormalize the helicity expressions to provide analogous winding experssions that are not dependent on the magnetic field strength. We then discuss properties of magnetic winding} and give demonstrations of how it can be used, in conjunction with magnetic helicity, to understand the topological evolution of magnetic fields.   

\section{Magnetic winding - what is it?}\label{wwii}
\subsection{Basic definitions}
In order to show how winding can describe the underlying topology of magnetic fields, we begin with some basic definitions. Consider a Euclidean space with the standard Cartesian basis $\{\ev_1,\ev_2,\ev_z\}$ (this notation highlights the importance of the $\ev_z$-direction that will become apparent later). Consider, further a horizontal plane $P$ with normal $\ev_z$. Let  $\bgamma: [a,b]\rightarrow\mathbb{R}^2\setminus\{\boldsymbol{0}\}$ be a suitably smooth (Lipschitz continuity will be assumed here) plane parametric curve on $P$ given in Cartesian coordinates by $\bgamma(t) = (x(t),y(t))$. In polar coordinates, the path can be written as
\be
\bgamma(t) = r(t)(\cos\theta(t),\sin\theta(t)).
\en
The distance from the (chosen) origin is $r(t)=\sqrt{x(t)^2+y(t)^2}$ and the angle $\theta(t)$ can be multivalued for any $t$ (by multiples of $2\pi$). In order to define $\theta(t)$ to be a unique function, we first choose a reference angle $\theta_a$ (for $t=a$). It can be checked that
\be
\frac{\d\theta}{\d t} = \dot{\theta}(t) = \frac{-y(t)\dot{x}(t)+x(t)\dot{y}(t)}{x(t)^2+y(t)^2}.
\en 
Therefore, the angle function with initial condition $\theta_a$ and the above derivative is
\be\label{theta}
\theta(t) = \theta_a + \int_a^t\frac{-y(u)\dot{x}(u)+x(u)\dot{y}(u)}{x(u)^2+y(u)^2}\,\d u.
\en
%It can be shown that $\theta(t)$ as defined in equation (\ref{theta}) is unique and the only continuous function with $\theta(a)=\theta_a$ (ref) [{\color{red} I'm not sure I understand this statement the to me point is the difference is$\theta(b)-\theta(a)$ does not depend on he value of $\theta(a)$ for a given curve, given this definition.}] 
With the angle function $\theta(t)$, the total (signed) change in angle of the path $\bgamma$ is $\theta(b)$-$\theta(a)$. The \emph{winding number} can now be defined as the net number of times $\bgamma$ goes around the origin, where anti-clockwise motion is positive and clockwise motion negative. The winding number of $\bgamma$ around the origin is thus given by
\be
\cl(\bgamma,\boldsymbol{0}) = \frac{1}{2\pi}\int_{\bgamma}\frac{-y\d{x}+x\d{y}}{x^2+y^2} = \frac{1}{2\pi}\int_a^b\frac{\d\theta(\bgamma)}{\d t}\,\d t.
\en
What we have described here is a standard definition of winding about a point (taken as the origin here although translating to another point is trivial).  Now let us consider the net winding between \emph{two} paths, rather than a path relative to a specific point. 

\begin{figure}
 \centering
\subfloat[]{\includegraphics[width=4cm]{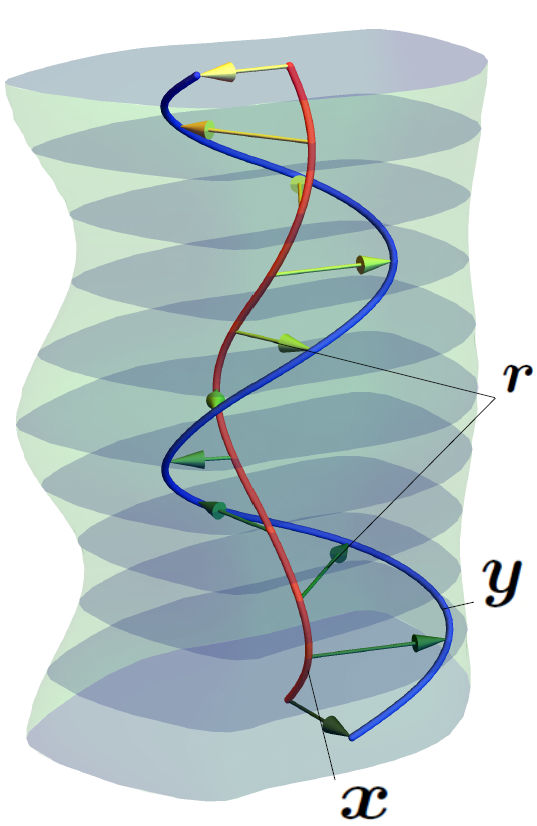}}\quad\quad\quad\subfloat[]{\includegraphics[width=7cm]{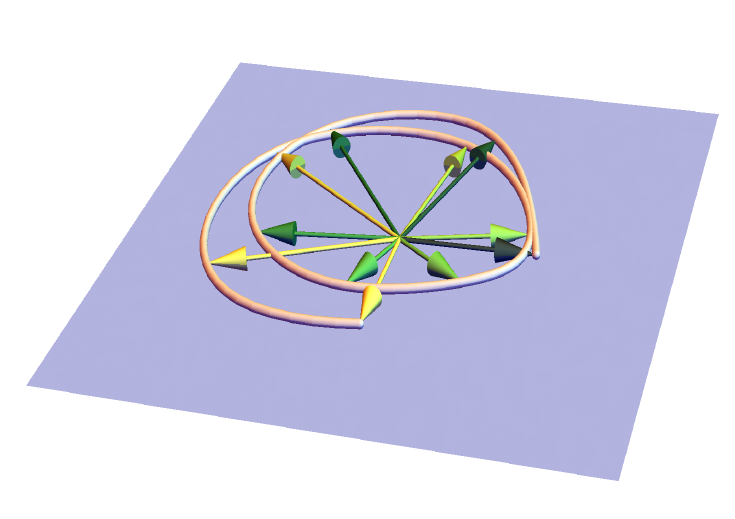}}
\caption{ (a) A domain $\Omega$ consisting of stacked horizontal slices that are orthogonal to $\ev_z$. Two curves (field lines) $\xv$ and $\yv$ are shown. Some slices are shown together with horizontal vectors $\rv$, at each height, connecting $\xv$ and $\yv$. Magnetic field in $\Omega$ is tangent to the  side of the domain but can be arbitrary on the top and bottom horizontal boundaries. For the calculation of helicity and winding fluxes, each slice can represent the boundary at a different instant of time rather than a location in space. (b) The $\rv$ vectors projected onto one horizontal plane. The resulting parametric curve allows for the calculation of  the pairwise winding of two field lines.}\label{fig_domain}
\end{figure}

Consider two parametric curves $\bgamma,\bgamma': [a,b]\rightarrow\mathbb{R}^2$ on $P$ that are distinct for all $t\in [a,b]$. For each $t$, set $\rv(t) = \bgamma(t)-\bgamma'(t)$. Therefore, the winding number between two curves is just the net rotation of $\rv$ about the origin, i.e.
\be\label{w2c}
\cl(\bgamma,\bgamma')=\cl(\rv,\boldsymbol{0}) = \frac{1}{2\pi}\int_{a}^{b}\frac{\d\theta(\rv)}{\d t}\,\d t.
\en
It is now straightforward to generalize the winding in equation (\ref{w2c}) to describe the net rotation of two three-dimensional curves (which we will shortly identify with magnetic field lines) about each other. Consider two smooth and distinct parametric curves $\xv, \yv:[0,h]\rightarrow\mathbb{R}^3$  that are monotonically increasing the $\ev_z$-direction, as shown in Figure \ref{fig_domain}(a). The vertical range of the curves is $[0,h]$.  For each $z\in[0,h]$, we define the horizontal vector $\rv(z)=\xv(z)-\yv(z)$. We then project each $\rv(z)$ for all $z\in[0,h]$ onto one horizontal plane. Due to the smoothness and continuity of the curves $\xv$ and $\yv$, the resulting projection of $\rv(z)$ for $z\in[0,h]$ is a suitably smooth planar parametric curve, as indicated in Figure \ref{fig_domain}(b). Hence, we can make use of the expression for winding given in equation (\ref{w2c}) and write
\be\label{wind1}
\cl(\xv,\yv) = \frac{1}{2\pi}\int_0^h\frac{\d}{\d z}\theta(\xv(z),\yv(z))\,\d z.
\en
Note that performing the integration in equation (\ref{wind1}) leads to
\be
\cl(\xv,\yv) = \frac{1}{2\pi}[\theta(\xv(h),\yv(h))-\theta(\xv(0),\yv(0))] + n,
\en
where $n\in\mathbb{Z}$ is the number of full rotations (in the sense described above) of the projected $\rv$ around the origin. If the positions of the curves at $z=0$ and $z=h$ remain fixed, then any smooth deformation of the curves $\xv$ and $\yv$ that does not result in cuts or reconnection, preserves the value of $n$ and, hence, $\cl$ \cite{berger2006writhe}. 
%\begin{figure}[h]
% \centering
%{\includegraphics[width=7cm]{winding.png}}
%\caption{The $\rv$ vectors projected onto one horizontal plane. The resulting parametric curve allows for the calculation of  the pairwise winding of two field lines.}\label{fig_domain}
%\end{figure}

If we now consider the curves $\xv$ and $\yv$ to follow magnetic field lines, $\cl$ represents a topological constraint on the magnetic field. Just as the Gauss linking number describes the pair-wise linkage of closed curves, the winding number $\cl$ describes the pair-wise winding of two open curves.

From the above discussion, the winding represents a two-dimensional description of the topology of curves (as opposed to the inherently three-dimensional description of Gauss linkage). Although this is fine for curves monotonically increasing in the $\ev_z$-direction, more work is required to define an adequate measure of winding for curves that bend backwards and are not monotonically increasing in the $\ev_z$-direction. Berger and Prior \cite{berger2006writhe} found such a measure by splitting the curves into regions separated by turning points. Suppose $\xv$ and $\yv$ have $n$ and $m$ distinct turning points respectively, that is points where $\d x_z/\d z = 0$ $(\xv\cdot\ev_z=x_z)$ or  $\d y_z/\d z = 0$ $(\yv\cdot\ev_z = y_z)$. Now split $\xv$ into $n+1$ regions and $\yv$ into $m+1$ regions. In each region, curve sections $\xv_i$ and $\yv_j$ share a mutual $z$-range $[z_{ij}^{\rm min}, z_{ij}^{\rm max}]$. Hence, in each section, equation (\ref{wind1}) can be applied and the total winding can be written as
\be\label{wind2}
\cl(\xv,\yv) = \sum_{i=1}^{n+1}\sum_{j=1}^{m+1}\frac{\sigma(\xv_i)\sigma(\xv_j)}{2\pi}\int_{z_{ij}^{\rm min}}^{z_{ij}^{\rm max}} \frac{\d}{\d z}\theta(\xv_i(z),\yv_j(z))\,\d z,
\en
where $\sigma(\xv_i)$ is an indicator function marking where the curve section $\xv_i$ moves up or down in $z$, i.e.
\be\label{sigma}
\sigma(\xv_i) = \left\{\begin{array}{ccc}
1 & {\rm if} \quad & \d x_z/\d z > 0, \\
-1 & {\rm if} \quad &\d x_z /\d z < 0 \\
0 & {\rm if} \quad &\d x_z /\d z = 0. \end{array}\right.
\en
Again, if the curves are fixed on the horizontal boundaries at $z=0$ and $z=h$, and are deformed smoothly without cuts or reconnection, the generalized winding in equation (\ref{wind2}) is conserved.

It was noted in \cite{berger2006writhe,prior2014helicity} that if the domain containing closed curves is divided into horizonal slices, as we have described in this section, the Gauss linkage of these curves is equal to the winding given in equation (\ref{wind2}). Thus, winding can be considered to be a more extensive topological description of field line entanglement than the Gauss linking number. We will now describe how winding forms the fundamental topological description of the helicity of open magnetic fields.

\subsection{Winding gauge}
As mentioned earlier, relative helicity is a common measure of helicity for open magnetic fields. For a restricted (but still very general) domain, the magnetic helicity of open magnetic fields can be expressed without the need of a reference field \cite{prior2014helicity}. Consider a simply connected domain $\Omega$ constructed of ``stacked'' horizontal planes, as portrayed in Figure \ref{fig_domain}(a). Field lines can be connected to the top and bottom horizontal boundaries (or neither for closed field) and are tangent to the side boundaries. 
%\begin{figure}\label{fig_oo}
%
%  \centering
%  \begin{tikzpicture}
%\draw[thick,-latex] (0,0,0) -- (3,0,0) node[anchor=north east]{$y$};
%\draw[thick,-latex] (0,0,0) -- (0,4,0) node[anchor=north west]{$z$};
%\draw[thick,-latex] (0,0,0) -- (0,0,3) node[anchor=south]{$x$};
%
%\draw (-2,4) -- (-2,0) arc (180:360:2cm and 0.5cm) -- (2,4) ++ (-2,0) circle (2cm and 0.5cm);
%\draw[densely dashed] (-2,0) arc (180:0:2cm and 0.5cm);
%
%\draw[densely dashed] (-2,2.8) arc (180:0:2cm and 0.5cm);
%
%\draw[thick] (-2,2.6) arc (180:360:2cm and 0.5cm);
%
%\node at (2.8,2.8) [anchor=east]{$S_z$};
%\node at (-0.8,1.2) [anchor=east]{$\Omega$};
%\node at (-2,0) [anchor=east]{$z=0$};
%\node at (-2,4) [anchor=east]{$z=h$};
%     \end{tikzpicture}  
%  
%     \caption{Cylindrical domain $\Omega$. A horizontal slice at height $z$ is denoted $S_z$. The domain is shown as a circular cylinder but other uniform shapes (e.g. a cuboid) or a deformed cylinder of stacked slices would also be described by this general setup.}
%
%\end{figure}
This domain can also be turned into the half-space by pushing the top and side boundaries to infinity and assuming that the magnetic field decays suitably quickly with distance from its source on the bottom plane. 

Consider a particular gauge, known as the \emph{winding gauge}, 
\be
\Av^W =\frac{1}{2\pi}\int_{S_z}\Bv(y_1,y_2,z)\times \frac{\rv}{|\rv|^2} \d^2y,
\en 
where $\rv = (x_1-y_1, x_2-y_2,0)$ and $S_z$ is a horizontal surface at height $z$. $\Av^W$ is a suitable vector potential for $\Bv$ in $\Omega$. Also, the winding gauge satisfies $\nabla_{\perp}\cdot\Av^W=0$ on any surface $S_z$. This gauge can be thought of as a two-dimensional equivalent of the usual Coulomb gauge, which Moffatt used to show that the topological structure of closed-field helicity is described by Gauss linkage. As its name suggests, the winding gauge leads to a description of open-field helicity whose underlying topological structure is encoded in the winding of its field lines.

Prior and Yeates \cite{prior2014helicity} defined a \emph{winding helicity} which can be written as
\be\label{windh}
H^W(\Bv) = \int_{\Omega}\Av^W\cdot\Bv\d V = \frac{1}{2\pi}\int_0^h\int_{S_z}\int_{S_z}\frac{\d}{\d z}\theta(\xv,\yv)B_z(\xv)B_z(\yv)\d^2x\,\d^2y\,\d z.
\en
Notice from the right-hand side of equation (\ref{windh}) that the winding helicity is the average pair-wise winding of field lines (see equation (\ref{wind1})) \emph{weighted} by magnetic flux. Again, this topological description of open-field helicity is analogous to the closed-field case (see equations (\ref{gauss}) and (\ref{class})). For applications to be discussed later, it will be useful to consider magnetic helicity as having a geometric structure (field line curves) with a topological description (winding) ``clothed'' by magnetic field. 

There are immediate consequences for the value of helicity due to the combined topology-field strength description given above. First, the helicity can be zero even if field strength is large everywhere but the field line topology has an equal number of positive and negative windings. Second, even if a field has a highly complex (non-cancelling) topology, the value of helicity can be small if the field strength is weak. From these simple observations, it would be useful to calculate the field line topology as well as the helicity in order to get a better picture of the total magnetic field topology. We will return to this important point later.

A related invariant of ideal MHD is the \emph{field line helcity} \cite{berger1988energy} which, for a field line curve $\xv$ and general vector potential $\Av$, is written as
\be\label{flh_gen}
\flh(\Bv) = \int_{\xv}\Av\cdot\d \lv,
\en
By applying the winding gauge, we can show that the field line winding represents the average winding of all other field lines with the field line in question, weighted by magnetic flux. For a magnetic field in the domain $\Omega$, the field line helicity of a curve ${\xv}$, which can be split into $n$ monotonic sections  with ranges $z \in[z_i^{\rm min},z_i^{\rm max}]$ by $n-1$ vertical turning points, can be written as
\be
\label{eqn:flh}
\flh = \int_{\xv}\Av^W\cdot\d \lv =\sum_{i=0}^{n}\int_{z_i^{\rm min}}^{z_i^{\rm max}} \Av^W\cdot\frac{\Bv_i}{|B_{zi}|}\,\d z = \frac{1}{2\pi}\sum_{i=0}^{n}\int_{z_i^{\rm min}}^{z_i^{\rm max}}\int_{S_z}\frac{\d}{\d z}\theta(\xv_i,\yv)\sigma(\xv_i)B_z(\yv)\,\d^2y\,\d{z}.
\en

\subsection{Relative helicity}\label{sec_relhel}
There is always a trade-off in helicity calculations between generality and topological information. Although relative helicity (equation (\ref{hr})) can be written in any suitable gauge, it is difficult to attach any clear topological interpretation to a general formulation. Although the winding helicity is not general, in the sense that it is based on a specific gauge, it allows for a much deeper topological interpretation compared to a more general formulation. One clear connection between winding helicity and relative helicity is that in the domain $\Omega$, 
\be
H_R(\Bv,\Bv') = H^W(\Bv) - H^W(\Bv').
\en
A reference field $\Bv'$ can always be found so that $H^W(\Bv')=0$, thus equating the relative and winding helicities \cite{prior2014helicity}. 

The connection between relative helicity and winding helicity does not end there, however. The practical injection of relative helicity, as mentioned in the Introduction, is performed by integrating the input of helicity through a boundary in time. In ideal MHD, the time-integrated rate of change of the relative helicity of any magnetic field passing through a horizontal planar boundary $P$ (or $S_0$ in $\Omega$) in the time range $[0,T]$ can be made equivalent to the winding helicity given in equation (\ref{windh}).

It was shown in \cite{berger1984limit,pariat2005helicity} that the time-integrated input of relative helicity through $P$ {(from equation \ref{hr_rate})} can be written as
\be\label{hr_ti}
H_R = -\frac{1}{2\pi}\int_0^T\int_P\int_P\frac{\d}{\d t}\theta(\xv,\yv)B_z(\xv)B_z(\yv)\,\d^2x\,\d^2y\,\d t.
\en
{The `advantages' of equation (\ref{hr_rate}) that were mentioned earlier are now clear in equation (\ref{hr_ti}), which reveals that the time-integrated helicity flux measures the winding of field lines in time weighted by magnetic flux. Equation (\ref{hr_ti}) was the first expression found that directly connects helicity and winding \cite{berger1986braids}. With the introduction of the winding helicity above, comparing this with the relative helicity flux reveals that they have a common underlying topological structure.}

Upon inspection, equation (\ref{hr_ti}) can be transformed to equation (\ref{windh}) by identifying $z \leftrightarrow t$ and $[0,T]\leftrightarrow[h,0]$. To understand why the order of integration needs to be swapped, consider magnetic field (e.g. a flux tube), initially below $P$, which then rises rigidly through $P$ until a time $T$. At time $T$, the helicity of the magnetic field would be integrated, using equation (\ref{windh}), from $z=0$ up to the maximum height of the emerged magnetic field, at $z=h$ say. Using equation (\ref{hr_ti}), however, the integration is performed in reverse, since the slice corresponding to the top of the magnetic region is counted at $P$ first (it is the first to pass through $P$). This process continues until the last slice, which corresponds to $z=0$ at time $t=T$. 

The identification made above can also be considered as a mapping between $\Omega$ and $\Omega_t$, where the latter domain is the same as the former but with the $z$-coordinate replaced by time. That is, $\Omega_t$ comprises of stacks of $P$ at different times in the range $[0,T]$. If a magnetic field passes through $P$ in a complex manner (e.g. a mixture of emergence and submergence in different parts of $P$), the field integrated in $\Omega_t$ will no longer be equivalent to the magnetic field in $\Omega$, although the winding helicity in $\Omega$ will still be equal to the time-integrated relative helicity in $\Omega_t$. As a simple example, consider the emergence and then complete submergence of a magnetic loop (e.g. magnetic field in a semi-torus). At the end of the submergence, there is no magnetic field in $\Omega$ and so the winding helicity is trivially zero. In $\Omega_t$ the emergence and submergence of every field line creates closed loops which are unlinked with all other loops. Hence, the helicity is also zero in this case.

\subsection{Separating winding and helicity}\label{sec_sepwind}

So far, we have demonstrated that, through the use of the winding gauge, the helicity of open magnetic fields can be interpreted as the field line winding weighted by the magnetic flux. In ideal MHD, however, the winding of the field lines themselves is also an invariant to deformations which vanish on the boundaries. It, therefore, makes sense to seek a purely topological measure of the field lines that is independent of the field strength. In studying flux emergence, \cite{prior2019interpreting,mactaggart2020helicity} defined a time-integrated magnetic winding flux $L_R$\footnote{The symbol $L$ is used for winding (and $L_R$ for winding flux) due to its connection to $Lk$ (see \cite{berger2006writhe}). Capital Roman letters refer to total measures of helicity and winding and calligraphic letters refer to the field line versions of these quantities. All further references to magnetic winding correspond to the definitions in this subsection.}
\be\label{lr_ti}
L_R = -\frac{1}{2\pi}\int_0^T\int_P\int_P\frac{\d}{\d t}\theta(\xv,\yv)\sigma(\xv)\sigma(\yv)\,\d^2x\,\d^2y\,\d t.
\en
Here, the $B_z$ terms from equation (\ref{hr_ti}) have been replaced with corresponding indicator functions $\sigma$. Equation (\ref{lr_ti}) provides topological information about the magnetic field that is unbiased by the magnetic field strength 

We can also define the magnetic winding in $\Omega$ by constructing a purely geometric version of the winding gauge. Consider the tangent vector $\Tv=\Bv/|\Bv|$ to the magnetic field and let $\Gv = \Tv/|T_z|=\sigma\d \xv/\d z$. We define $\Cv$ to be a geometric analogue of the winding gauge as
\be
\Cv = \frac{1}{2\pi}\int_{S_z}\Gv(\yv)\times\frac{\rv}{|\rv|^2}\,\d^2 y.
\en 
Analogous to the calculations in \cite{prior2014helicity}, it can be shown that
\be
\label{eqn:windprod}
\Cv\cdot\Gv = \frac{1}{2\pi}\int_{S_z}\frac{\d}{\d z}\theta(\xv,\yv)\sigma(\xv)\sigma(\yv)\,\d^2y.
\en
Integrating this quantity over the volume of $\Omega$ gives
\be\label{winding_omega}
L = \frac{1}{2\pi}\int_0^h\int_{S_z}\int_{S_z}\frac{\d}{\d z}\theta(\xv,\yv)\sigma(\xv)\sigma(\yv)\,\d^2x\,\d^2y\,\d z.
\en
Similar to the field line helicity, we can also define a field line winding
\be\label{field_line_wind}
\flw(\Bv) = \int_{\xv}\Cv\cdot\d\lv.
\en
The field line winding of curve $\xv$, which can be split into $n$ monotonic sections  with ranges $z \in[z_i^{\rm min},z_i^{\rm max}]$ by $n-1$ turning points,  can be written as
\be
\flw = \frac{1}{2\pi}\sum_{i=0}^{n}\int_{z_i^{\rm min}}^{z_i^{\rm max}}\int_{S_z}\frac{\d}{\d z}\theta(\xv_i,\yv)\sigma(\xv_i)\sigma(\yv)\,\d^2y\,\d z,
\en
which is just the average pair-wise winding of all field lines in $\Omega$ with $\xv$. {Prior and Yeates \cite{prior2018quantifying} showed that $\flw$ can be used to topologically categorize braided magnetic fields, where the field line helicity cannot. This was used to quantify the changing connectivity of magnetic flux rope experiments carried out at the UCLA basic plasma facility.}

\section{Winding - what is it good for?}
In this section we will demonstrate that magnetic winding can provide different and more detailed information about field line topology than the helicity.

\subsection{Resistive magnetostatics}

\subsubsection{Linear force-free fields}
%We now focus on the circumstances under which the magnetic helicity and the winding behave differently. 
The magnetic induction equation can be written as 
\begin{equation}\label{induction}
\frac{\partial\Bv}{\partial t} = \nabla \times ({\uv}\times\Bv)-\nabla\times(\eta\nabla\times\Bv), 
\end{equation}
where $\eta$ is the magnetic diffusion. Only the second term on the right-hand side of equation (\ref{induction}) can change the values of $H^W$ and $L$ (and their associated field line partners $\flh$ and $\flw$) so we will focus on this term by considering the quasi-static evolution with $\uv=\boldsymbol{0}$. We will demonstrate that different components affect helicity and winding in distinct ways, thus emphasising the importance of separating winding from helicity. For the sake of simplicity, we assume that $\eta$ is constant in what follows.

 We begin with one of the simplest possible current carrying magnetic fields, a linear force-free field, that satisfies $\nabla\times\Bv = \alpha\Bv$ for constant $\alpha$, in the domain $\Omega$ (we set $\mu_0=1$ in this work for convenience).  The top and bottom boundaries are fixed or periodic. The magnetic field remains tangent to  the side boundary. The following theorem tells us that only the helicity decays for such fields.
 
\begin{theorem}\label{thrm1}
The winding helicity and field line helicity of a linear force-free field in $\Omega$ subject to constant magnetic diffusion $\eta$ and no flow obey
\be
H^W[\Bv(t)]=H^W[\Bv(0)]\exp(-2\alpha^2\eta t)\quad{\rm and}\quad\flh[\Bv(t)]=\flh[\Bv(0)]\exp(-\alpha^2\eta t),
\en
where $\alpha$ is the force-free constant. Under the same assumptions, the winding quantities $L[\Bv(t)]$ and $\flw[\Bv(t)]$ are constant. 
\end{theorem}

\begin{proof}
We start by proving that  $L[\Bv(t)]$ and $\flw[\Bv(t)]$ are constant under such an evolution. The quasistatic induction equation is
\be\label{ind_static}
\frac{\partial\Bv}{\partial t} = -\nabla\times(\eta\nabla\times\Bv).
\en
Jette \cite{jette1970force} (see also \cite{losurdo1979static}) proved that in resistive magnetohydrostatics, the only force-free fields $\Bv$ that remain force-free in time are those with constant $\alpha$, i.e. linear force-free fields. Using the linear force-free equation,
\be\label{dhwdt}
\nabla\times(\eta\nabla\times\Bv) = \nabla\times(\eta\alpha\Bv)=\alpha^2\eta\Bv.
\en
From this result, the magnetic field behaves as
\be\label{jette}
\Bv(t) = \Bv(0)\exp(-\alpha^2\eta t),
\en
where $\Bv(0)$ is the magnetic field at the start of the quasistatic decay . Therefore, the field line structure of $\Bv$ remains constant for all $t>0$, that is $\Tv(t) = \Tv(0)$, and so $L[\Bv(t)]$ and $\flw[\Bv(t)]$ are fixed in time. Then from equations (\ref{windh}) and (\ref{eqn:flh}) and the fact that 
\be
B_z=B_{z}(0)\exp(-\alpha^2\eta t),
\en
we have
\be
H^W[\Bv(t)]=H^W[\Bv(0)]\exp(-2\alpha^2\eta t)\quad{\rm and}\quad \flh[\Bv(t)]=\flh[\Bv(0)]\exp(-\alpha^2\eta t).
\en
\end{proof}
This particularly simple situation shows that even when there is no flow, the helicity and the winding can behave differently, despite being so intimately related. Therefore, each quantity can provide \emph{different} information on the overall behaviour of the evolving magnetic field.  

%\begin{remark}
%In a magnetically bounded domain, Theorem \ref{thrm1} holds when using the Coulomb gauge instead of the winding gauge. The proof is almost identical to that above and is omitted for brevity.
%\end{remark}

\subsubsection{General differences in the decay of $H^W$ and $L$}
We now explore more general differences between the winding and the helicity and relate them to specific physical properties of the magnetic field. To do so we write the curl of the magnetic field as the sum of force-free and Lorentz force ($\Fv_l$) generating components. Using Amp\`{e}re's law $\Jv=\nabla\times \Bv$ and $\Fv_l =\Jv\times\Bv$, it can be checked that
\begin{equation}
\label{curlexp1}
\nabla \times \Bv =\Bv\times\frac{\Fv_l}{\vert\Bv\vert^2}+ \alpha \Bv,\quad \alpha = \frac{(\nabla\times\Bv )\cdot \Bv}{\vert\Bv\vert^2}.
\end{equation}
In the second term, $\alpha$ represents the component of the axial current weighted by the field strength. The parameter $\alpha$ also represents a topological quantity: the mean twisting of the field around the field line passing locally through the point at which the field $\Bv$ is anchored \cite{berger2006writhe}. In the first term, the vector $\Fv_l/\vert\Bv\vert^2$ points along the direction of the Lorentz force but with a magnitude that is the ratio of the Lorentz force strength to the square of the field strength. Thus we can write (\ref{curlexp1}) as
\begin{equation}
\label{curlexp2}
\nabla \times \Bv =\lambda\Bv_{f\perp} + \alpha {\Bv},
\end{equation}
where
\begin{equation}
\nvec  = \Bv\times \hat{\Fv_l},\quad\hat{\Fv_l}=\frac{\Fv_l}{\vert \Fv_l\vert},
\end{equation}
and
\begin{equation}
\lambda  = \frac{\nabla\times\Bv\cdot{\nvec}}{\vert \Bv\vert^2}.
\end{equation}
From a geometric perspective, $\lambda$  measures the rotation of the magnetic field around the direction of the vector $\Bv\times \hat{\Fv_l}$, a vector normal to both the Lorentz force and magnetic field. If we specify the function $\lambda$ then the Lorentz force can be written as 
\begin{equation}
\Fv_l= {\lambda}\nvec\times \Bv = \lambda\vert \Bv\vert^2\hat{\Fv}_l,
\end{equation}
and $\lambda$ can be seen to represent the relative magnitude of the Lorentz force to the magnetic field strength. Thus the representation of the magnetic field's varying local geometry through (\ref{curlexp2}) has two scalar parameters, $\lambda$ and $\alpha$, which represent, respectively, the relative magnitude of the Lorentz force and axial current and hence measure their relative effect on the local geometry of the magnetic field through its curl. 
%We argue that this is a meaningful decomposition as the Lorentz  is the component which acts in the momentum equation of the full MHD system whilst the axial current is a commonly used diagnostic of magnetic flux rope structure (e.g. \cite{guo2013twist,liu2016structure,guo2017magnetic}).

As shown in the Appendix, the (quasistatic) induction equation can be written in terms of components parallel and perpendicular to the magnetic field,
\begin{align}
\label{eqn:bderiv}
\frac{\partial\Bv}{\partial t}  &=-\eta\left(\Cv_{\perp} + \Cv_{\parallel}\right),\\
\label{eqn:bderiv_par}\Cv_{\parallel} &= \left( \nabla\lambda\cdot \hat{\Fv}_l+\alpha^2+\lambda\omega_b\right)\Bv,\\
 \label{eqn:bderiv_perp}\Cv_{\perp}& = -(\nabla\lambda\cdot\Bv)\hat{\Fv}_l+ \alpha\lambda\nvec +  \nabla\alpha\times{\Bv} + \lambda\omega_f \hat{\Fv}_l+\lambda\omega_{fl}\Bv_{f\perp},
\end{align}
where $\omega_b $ represents the rotation of the Lorentz force vector around a field line, $\omega_f$
represents the rotation of the magnetic field $\Bv$ around the direction of the Lorentz force and $\omega_{fl}$
represents the rotation of the pair $(\Bv,\hat{\Fv}_l)$ around the direction of the field  $\nvec$. The expressions for these scalars can be found in the Appendix.

Before discussing the interpretation of individual terms of $\Cv_{\perp}$ and $\Cv_{\parallel}$, we highlight the importance of decomposing the vector $\partial \Bv/\partial t$ into components parallel and perpendicular to the magnetic field. First, the topology of the field, which can be found from the unit tangent vector ${\Tv} = \Bv/\vert \Bv\vert$ or from the vector $\Gv$ (the $z$-derivative if the curve $\xv$), is only affected  by $\Cv_{\perp}$, i.e.
\begin{equation}
\label{eqn:topderiv}\frac{\partial{\Tv}}{\partial t} =-\eta\frac{ \Cv_{\perp}}{\vert {\bf B}\vert},\quad \frac{\partial \Gv}{\partial t}= -\eta\sigma\frac{\Cv_{\perp}B_z- (\Cv_{\perp}\cdot\ev_z)\Bv}{B_z^2},
\end{equation}
where the second expression uses the fact we can write $\Gv = \sigma\Bv/B_z$. By contrast, the magnetic field's magnitude is only changed by $\Cv_{\parallel}$,
\begin{equation}
\pder{\vert \Bv\vert}{t} =-\eta\frac{\Cv_{\parallel}\cdot \Bv}{\vert \Bv\vert}.
\end{equation}
These properties indicate that the change in magnetic winding is \emph{only} affected by $\Cv_{\perp}$, whilst the change in helicity depends on both $\Cv_{\perp}$ and $\Cv_{\parallel}$. To give compact representations of their changes we define the operator 
\begin{equation}
\Wv(\Nv) = \frac{1}{2\pi}\int_{S_z}\Nv(\yv)\times\frac{\rv}{\vert\rv \vert}\,\d^2 y,
\end{equation}
for any vector field $\Nv$, {e.g.} $\Cv= \Wv(\Gv)$ and $\Av^W= \Wv(\Bv)$ are examples. Then the change in helicity is
\begin{align}
\pder{H^W}{t} &=\int_{\Omega}\pder{\Bv}{t}\cdot \Wv(\Bv)\,\d V + \int_{\Omega}\Bv\cdot \Wv\left(\pder{\Bv}{t}\right)\,\d V, \nonumber\\
&=-\eta\left[ \int_{\Omega}(\Cv_{\perp} + \Cv_{\parallel})\cdot \Wv(\Bv)\,\d V + \int_{\Omega}\Bv\cdot\Wv\left( \Cv_{\perp} + \Cv_{\parallel}\right)\,\d V\right],
\end{align}
where we have used (\ref{windh}) in conjunction with (\ref{eqn:bderiv}). Similarly, using (\ref{eqn:windprod}) in conjunction with (\ref{eqn:topderiv}) we find
\begin{align}
\pder{L}{t}=-\eta\left[ \int_{\Omega}\sigma(\xv)\left(\frac{\Cv_{\perp}B_z- (\Cv_{\perp}\cdot\ev_{z})\Bv}{B_z^2}\right)\cdot \Wv(\Gv)\,\d V \right. \nonumber\\
\left. + \int_{\Omega}\Gv\cdot\Wv\left( \sigma(\yv)\frac{\Cv_{\perp}B_z- (\Cv_{\perp}\cdot\ev_{z})\Bv}{B_z^2}\right)\,\d V\right].
\end{align}
Using (\ref{eqn:flh}) and (\ref{field_line_wind})  in conjunction with (\ref{eqn:bderiv}) and (\ref{eqn:topderiv}), the changes in field line helicity and field line winding are
\begin{align}
\pder{{\wh}}{t} &=-\eta\left[\sum_{i=1}^{n} \int_{z_i^{\rm min}}^{z_i^{\rm max}}\sigma(\xv_i)\left(\frac{\Cv_{\perp i}B_{zi}- (\Cv_{\perp i}\cdot\ev_z)\Bv_i}{B_{zi}^2}\right)\cdot\Wv(\Bv)\,\d z\right. \nonumber \\
& \left. + \sum_{i=1}^{n} \int_{z_i^{\rm min}}^{z_i^{\rm max}}\Gv_i\cdot \Wv(\Cv_{\perp} + \Cv_{\parallel})\,\d z\right],
\end{align}
\begin{align}
\pder{{\cal L}}{t} &= -\eta\left[\sum_{i=1}^{n}\int_{z_i^{\rm min}}^{z_i^{\rm max}}\sigma(\xv_i)\left(\frac{\Cv_{\perp i}B_{zi}- (\Cv_{\perp i}\cdot\ev_z)\Bv_i}{B_{zi}^2}\right)\cdot\Wv(\Gv)\,\d z\right. \nonumber\\ 
&\left.+ \sum_{i=1}^{n}\int_{z_i^{\rm min}}^{z_i^{\rm max}}\Gv_i\cdot \Wv\left( \sigma(\yv)\frac{\Cv_{\perp}B_z- (\Cv_{\perp}\cdot\ev_z)\Bv}{B_z^2}\right)\right]\,\d z. 
\end{align} 
Note that there is no contribution from derivatives of the integral boundaries  (which would result from the Leibniz rule) as contributions cancel from connected integrals.

With these expressions in place, we can now review what physical effects will alter these quantities. By considering equations (\ref{eqn:bderiv}) to (\ref{eqn:bderiv_perp}),  the terms which change \emph{only the field strength} and, hence, only affect the helicity are
\begin{enumerate}
\item{The square of the strength of the axial current  relative to the magnetic field strength: $\alpha^2$.}
\item{Gradients of the relative strength of the Lorentz force along the direction of the Lorentz force: ${\nabla\lambda\cdot \hat{\Fv}_l}$.
}\item{
The rate of rotation of the Lorentz force around the curve:
\begin{equation}
\lambda\omega_b=\lambda\frac{\nabla\times \Bv_{fl}\cdot\Bv }{\vert \Bv \vert^2}.
\end{equation}
(since $\nvec$ and $\Fv$ form an orthogonal pair spanning the plane normal to $\Bv$) 
}
\end{enumerate} 
Terms (i-ii) represent variations of the relative strength of the Lorentz force and the axial current. Term (iii) accounts for the rotating geometry of the Lorentz force vector along a field line.
Terms which change \emph{only the field topology}, and which affect the winding, are
\begin{enumerate}
\item{Gradients of the axial current normal to the magnetic field: $\nabla\alpha\times {\Bv}$.}
\item{The product of the relative strength of the axial current and Lorentz force: $\alpha\lambda$. That is to say if both axial current and Lorentz forces are present there must be some change in topology.}
\item{Gradients of the relative strength of the Lorentz force along the direction of the magnetic field: ${\nabla\lambda\cdot \Bv}$.}
\item{
The rotation of field $\Bv$ around the direction of the Lorentz force:  $\lambda\omega_f$. 
}
\item{The rotation of the pair $(\Bv,\hat{\Fv}_l)$ around $\nvec$: $\lambda\omega_{fl}$.}
\end{enumerate}
Terms (i-iii) represent variations of the relative strength of the Lorentz force and the axial current. Terms (iv-v) account for the varying (relative) geometry of the fields $\Bv$ and $\Fv$. These last two components are difficult to visualize but, as we shall see in the following toy examples, are potentially significant {in magnetic flux ropes}.
\subsubsection{Toy examples}
We consider  simple toy examples of a  radially symmetric magnetic field in a cylindrical geometry ({i.e.} a simple flux rope model).  The field is assumed to take the form
\begin{equation}
\label{eq:simpfield}
\Bv = B_{\theta}(r)\ev_{\theta} + B_z(r)\ev_{z}.
\end{equation}
where $(r,\theta,z)$ is a cylindrical coordinate system with unit vectors $\{\ev_{r},\ev_{\theta},\ev_{z}\}$, $z\in[0,1]$ and $r\in[0,r_m]$. Due to the cylindrical symmetry and the fact the magnetic field has no $\ev_{r}$-component, the field lines are helical curves lying on concentric cylinders of radius $r$. The ratio $B_z/B_{\phi}$ is their helical pitch. Thus we can study all field line topological quantities as functions of $r$.

The curl of (\ref{eq:simpfield}) is
\begin{align}
\label{eqn:simpcurl}
\nabla \times \Bv &= -  \deriv{B_z}{r}\ev_{\theta} + \left(\deriv{B_{\theta}}{r} +\frac{B_{\theta}}{r} \right)\ev_{z}.
\end{align}
Since this quantity has no $\ev_{r}$-component, the Lorentz force must be directed along $\ev_{r}$-direction, {i.e.} $\hat{\Fv_l}=\ev_{r}$. Thus $\nvec = \Bv\times\ev_{r}= B_z\ev_{\theta} - B_{\theta}\ev_{z}$. Then, if we specify the functions $ \lambda(r)$ and $\alpha(r)$, we can equate equations (\ref{eqn:simpcurl}) and (\ref{curlexp2}) to obtain the following ordinary differential equations 
\begin{equation}
\deriv{B_z}{r}+\alpha B_{\theta}+\lambda B_z=0,\quad \deriv{B_{\theta}}{r} +\frac{B_{\theta}}{r}-\alpha B_z+\lambda B_{\theta}=0,
\end{equation}
which, if solved, give the field components $B_z$ and $B_{\theta}$ with the required $\lambda$ and $\alpha$ behaviour. {In this study we consider the initial conditions $B_{\theta}(0)=0$ and $B_z(0)=1$. The condition on $B_{\theta}$ means that $\lambda$ must be zero at $r=0$ for a valid solution.} 

Using the fact that 
\begin{equation}
\nabla \times \nvec =\left(\alpha-\frac{2 B_{\theta}B_z}{r B^2}\right)\Bv_{f\perp} + \left(-\lambda+\frac{B_z^2-B_{\theta}^2}{r B^2}\right)\Bv,
\end{equation}
we have, using equation (\ref{eqn:bderiv}),
\begin{equation}
\label{dbdtsym}
\pder{\Bv}{t} = -\eta\left[\left(2\alpha\lambda - \deriv{\alpha}{r} -\lambda \frac{2 B_{\theta}B_z}{r B^2}\right)\nvec + \left(\deriv{\lambda }{r} -\lambda ^2+\alpha^2 + \lambda \frac{B_z^2-B_{\theta}^2}{r B^2}\right)\Bv \right].
\end{equation}
We will consider cases where the field strength $|\Bv|$ and its axial twisting $\alpha$ are $\mathcal{O}(1)$. We further assume that  $\lambda\ll \mathcal{O}(1)$ but that its gradient $\d\lambda/\d r$ is $\mathcal{O}(1)$. With these assumptions, we can approximate the behaviour away from the core as 
\begin{equation}
\label{dbdtsymsmalllorentz}
\pder{\Bv}{t} \approx -\eta\left[- \deriv{\alpha}{r}\nvec + \left(\deriv{\lambda}{r}+\alpha^2\right)\Bv \right].
\end{equation}
From this we can see that an axial current through $\alpha$ and a Lorentz force gradient though the derivative $\d \lambda/\d r$ will tend to alter the magnetic field strength. The field line topology, however, can only change if there are gradients in the axial current. The terms 
\begin{equation}
 \lambda \frac{2 B_{\theta}B_z}{r B^2}\,\mbox{ and }\, \lambda \frac{B_z^2-B_{\theta}^2}{r B^2},
\end{equation}
can be prominent near the core centre, even if $\lambda$ is small {(the condition that $\lambda(0)=0$ means the second term is bounded)}. This discussion highlights the fact that small and local variations in the geometry of the magnetic field and the Lorentz force can lead to significant changes in helicity and winding, even if the magnitude of the Lorentz force is small. Since the the helicity and winding quantities involve integration over the plane, they can have a non-local effect on field lines away form the core.

\subsubsection{Force-free example}

\begin{figure}
\subfloat[]{\includegraphics[width=7cm]{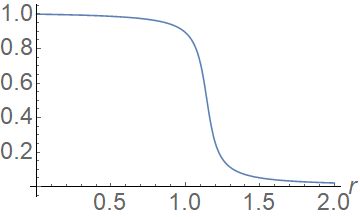}}\quad\subfloat[]{\includegraphics[width=8cm]{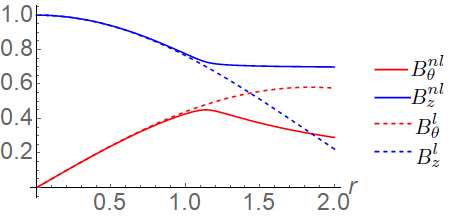}}
\caption{\label{fffr} Distributions of quantities associated with the fields $\Bv^l$ and $\Bv^{nl}$. (a) The axial twisting distribution $\alpha(r)$ for $\Bv^{nl}$. (b) The $B_z$ and $B_{\theta}$ components of the fields $\Bv^l$ (dashed) and $\Bv^{nl}$ (solid).}
\end{figure}

\begin{figure}
\subfloat[]{\includegraphics[width=7cm]{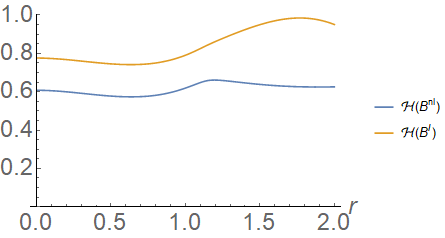}}\quad\subfloat[]{\includegraphics[width=7cm]{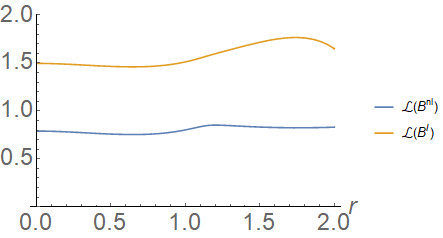}}\quad  \subfloat[]{\includegraphics[width=7cm]{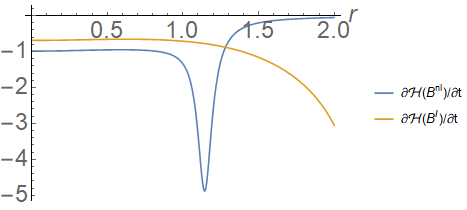}}\quad \subfloat[]{\includegraphics[width=7cm]{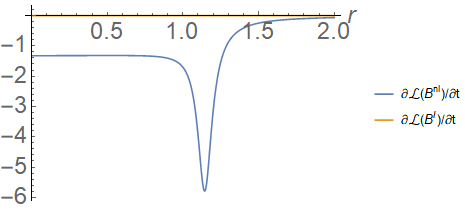}} 
\caption{\label{fig:forcefreedists} Field line helicity and winding distributions and their temporal changes for force-free fields. (a) ${\cal H}(r)$  for the fields $\Bv^l$ and $\Bv^{nl}$. (b) ${\cal L}(r)$  for the fields $\Bv^l$ and $\Bv^{nl}$. (c) $\partial {\cal H}/\partial t$  for the fields $\Bv^l$ and $\Bv^{nl}$, as a function of $r$. (d) $\partial {\cal L}/\partial t$ for the fields $\Bv^l$ and $\Bv^{nl}$, as a function of $r$. {Note that $\partial {\cal L}(\Bv^l)/\partial t=0$ for this case, as expected from Theorem \ref{thrm1}.}}
\end{figure}

We present some numerical results to highlight the effect of varying the $\alpha$ and $\lambda$ functions. First we consider force-free fields ($\lambda=0$).  We consider two fields on a domain $r\in[0,2]$
, a linear force-free field (constant $\alpha$) $\Bv^l$ and a nonlinear force-free field $\Bv^{nl}$ with a decaying axial current, specifically
\begin{equation}
\alpha(r) = \frac{\tan ^{-1}(20-17.5 r)+\frac{\pi }{2}}{\frac{\pi }{2}+\tan ^{-1}(20)}.
\end{equation}
This distribution is shown in Figure \ref{fffr}(a). It has a region at the magnetic field's core which has almost constant axial current which then drops of sharply with a significant gradient.  The magnetic field components for the two cases are shown in Figure \ref{fffr}(b). The major feature is the expected decay of the $B_{\theta}$ component for the field $\Bv^{nl}$ .

The field line helicities ${\cal H}(r)$  for both fields are shown in Figure \ref{fig:forcefreedists}(a). In both cases there is some variation but only as a relatively small percentage of the total value. There is less field line helicity in the nonlinear field as its twist decays where the linear case does not. The field line winding distributions, shown in Figure \ref{fig:forcefreedists}(b), are qualitatively very similar.  

In the following we set $\eta=1$. The changes $\partial {\cal H}/\partial t$  are shown in Figure \ref{fig:forcefreedists}(c) and the distributions differ significantly for both fields. In the inner core, $ r\in[0,0.8]$, where the fields are very similar, the distributions for the linear and nonlinear cases show a (radially) constant rate of loss. After this, the gradient $\d \alpha/\d r$ drives a rapid sipke in the decay of the nonlinear field's ${\cal H}$ distribution, resulting from the changing field topology related to the helical relaxation of the field. This decay then drops to zero as the twisting  decays. In the linear case there is no such spike but a steady increase towards the edge of the cylinder. This behaviour is due to the fact that there is a continual $\alpha^2$ decay and the fact that the length of the field lines increases towards the edge of the field - since $B_z$ decreases and $B_{\theta}$ increases (Figure \ref{fffr}(b)) the field line curves become increasingly tightly-coiled helices. Intriguingly, despite the fact that the two fields' ${\cal H}$ distributions are qualitatively similar, their temporal evolution is not. 

In Figure \ref{fig:forcefreedists}(d) we see the change in winding  $\partial {\cal L}/\partial t$, qualitatively, nearly identical to $\partial \wh/\partial t$ for the nonlinear field, owing to the dominance of the topology-changing $\d \alpha/\d t$ gradient. There is no change in the linear force-free field as expected from Theorem \ref{thrm1}.

\subsubsection{Adding a Lorentz force}
\begin{figure}
\subfloat[]{\includegraphics[width=7cm]{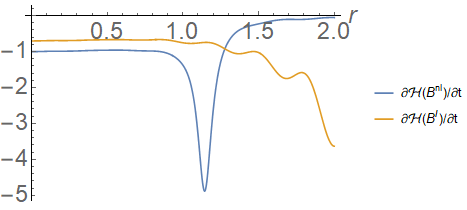}}\quad \subfloat[]{\includegraphics[width=7cm]{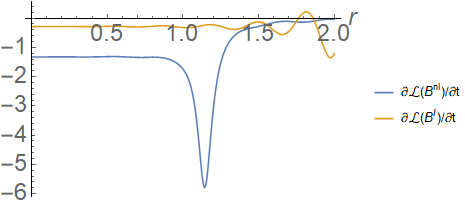}} 
\caption{\label{fig:lorentzdists} Temporal changes in  the field line helicity and winding distributions for fields with a weak and highly oscillatory Lorentz force. (a) Changes $\partial {\cal H}/\partial t$  for the fields $\Bv^l$ and $\Bv^{nl}$, as a function of $r$. (b) Changes $\partial {\cal L}/\partial t$ for the fields $\Bv^l$ and $\Bv^{nl}$ , as a function of $r$.}
\end{figure}
We now consider the effects of including a small Lorentz force, through $\lambda$, of the form
\begin{equation}
\lambda = 0.05\sin(20 r).
\end{equation}
This choice implies that the Lorentz force is between one and two orders of magnitude smaller than the magnetic field strength but that its gradient is $\mathcal{O}(1)$.

Figure \ref{fig:lorentzdists} displays $\partial \wh/\partial t$ and $\partial \cl/\partial t$  after re-solving for the fields $\Bv^l$ and $\Bv^{nl}$ with the above choice of $\lambda$. For brevity, we do not plot the distributions of ${\cal H}$ or ${\cal L}$ as they are very similar to those in Figures \ref{fig:forcefreedists}(a) and (b). The distributions for $\Bv^{nl}$ are also almost identical to those of the force-free case (these are dominated by the gradient in $\alpha$) and the major changes appear for $\Bv^l$. The $\partial \wh(\Bv^l)/\partial t$ distrubution in Figure \ref{fig:lorentzdists}(a) has the same general trend as that in Figure \ref{fig:forcefreedists}(c) but with oscillation added to it. The $\partial \cl(\Bv^l)/\partial t$ distribution in Figure \ref{fig:lorentzdists}(b) now exhibits a constant negative rate of field line winding in the inner region $r\in[0,1]$. By reperforming the calculations without the various terms it was established this change arises due to the  changing geometry and magnitude of the Lorentz force along the field through the term $2\alpha \lambda$ in equation (\ref{dbdtsym}), which leads to a significant contribution from the tightly coiled curves towards the flux rope's edge.
A pronounced oscillatory component of $\partial \cl(\Bv^l)/\partial t$ develops in the outer region $r>1$. This is due to the oscillation in the Lorentz force magnitude $\lambda$ and the effect is magnified as field lines become increasingly tightly coiled, to the point where there is a region where the winding is being increased rather than decaying. 

\subsubsection{General conclusions}
We have demonstrated that there can be significant differences between the evolution of the field line winding and helicity. Further, we have linked these differences to various physical properties of the field, specifically components and magnitudes of axial currents and the Lorentz force. We have also given some examples of how gradients in both the axial current and weak Lorentz forces might affect magnetic fields with flux rope-type geometries. Of course, our analysis is only strictly valid when the Lorentz force is suitably small and decoupled from the fluid pressure gradient and inertia terms. Such a situation would be valid, for example, in the solar corona (at least before the onset of significant plasma motion). For a discussion of the effects of non-magnetic terms in quasistatic decay, the reader is directed to \cite{low1974resistive}.

\subsection{Flux emergence and submergence}
The emergence of magnetic field into the solar atmosphere is one of the main drivers of solar activity. It has long been appreciated that helicity is important for solar eruptions, so understanding how topologically complex field emerges into the solar atmosphere provides us with information on what kinds of eruptions can form. Recently, MacTaggart and Prior \cite{prior2019interpreting, mactaggart2020helicity} performed a detailed analysis on how helicity and winding is transported into the solar atmosphere in flux emergence simulations. In order to understand what topological information is being transported by emergence, we now present a simple but informative model of flux emergence that illustrates one of the main results of \cite{prior2019interpreting, mactaggart2020helicity}.

\subsubsection{Magnetic field construction}
Emerging magnetic fields in the solar atmosphere have two footpoints (sunspots) which move apart from each other until a certain distance. This behaviour suggests that the basic geometry of the emerging magnetic field is toroidal \cite{hood2012flux}. To construct a toroidal flux tube, we need to define a magnetic field in a toroidal domain. That is, we wish to create a divergence-free field with a field line structure of our choosing in a toroidal domain. 

 To construct such a field, we will consider transforming a magnetic field defined in a cylinder to one in a semi-torus. The reason for this is that it is much easier to define magnetic flux tubes of arbitrary complexity in cylinders.  

 For a toroidal flux tube, the axis of the tube can be written as the parametric curve
 \be
 \rv(s) = -R\cos(s/R)\ev_1 + (R\sin(s/r)+z_0)\ev_z,
 \en
 where $s$ is the arclength along the tube axis, $R$ is the major radius of the torus and $z_0$ is the height at which the footpoints of the tube meet the photosphere. The unit tangent vector of $\rv(s)$ is $\Tv = \d\rv(s)/\d s$ and by taking the (unit) normal and binormal vectors, $\dv_1$ and $\dv_2$ say, we can define an orthonormal triad $\{\Tv,\dv_1,\dv_2\}$. This basis can be used to define a tubular coordinate system through the mapping
 \be
 \fv(s,x_1,x_2) = \rv(s) + x_1\dv_1 + x_2\dv_2,
 \en
 with
 \[
 \dv_1 = \cos(s/R)\ev_1 -\sin(s/R)\ev_z, \quad \dv_2 = \ev_2.
 \]
 The metric tensor can be written as
 \be
 g_{ij} = \frac{\partial\fv}{\partial q^i}\cdot\frac{\partial\fv}{\partial q^j},
 \en
 where $i,j=1,2,3$ with $q^1=s$, $q^2=x_1$ and $q^3=x_2$. The Jacobian of the mapping is given by

 \be
 \sqrt{g} = \sqrt{\det(g_{ij})} = \frac{R-x_1}{R}.
 \en
% With the basic geometry in place, we are now in a position to construct magnetic field models. The divergence of a vector field $\Cv$ can be written as
% \be
% \nabla\cdot\Cv = \frac{1}{\sqrt{g}}\frac{\partial}{\partial q^i}(\sqrt{g}C^i).
% \en
 If $\Bv$ represents a magnetic field, it must obey the solenoidal constraint, i.e.
 \be
 \frac{\partial}{\partial q^i}(\sqrt{g}B^i) = 0.
 \en
 It follows, therefore, that any magnetic field defined in a cylinder can be converted to one in terms of the above tubular coordinates by simply dividing the components by $\sqrt{g}$.
 
For our application here, we will consider a twisted flux tube that has a small concentration of even higher twist at the centre of the tube, localized at the apex. To construct such a magnetic flux tube, we first write the magnetic field as 
\be
\Bv = B_s\frac{\partial\fv}{\partial s} + B_1\frac{\partial\fv}{\partial x_1} + B_2\frac{\partial\fv}{\partial x_2}.
\en
Then, following \cite{prior2019interpreting}, the magnetic field with twist at its centre $(x_1,x_2)=(0,0)$ can be written as
\be\label{tube_field}
\Bv = \frac{b_0}{\sqrt{g}}\left[\frac{2}{w}\exp\left(-\frac{x_1^2+x_2^2}{w^2} -\frac{(z-z_0)^2}{l^2}\right)+t_w\right]\left(-x_2\frac{\partial\fv}{\partial x_1} + x_1\frac{\partial\fv}{\partial x_2}\right) + \frac{b_0}{\sqrt{g}}\frac{\partial\fv}{\partial s},  
\en
where $b_0$ is the axial field strength, $t_w$ is the background twist and $w$ and $l$ control the size of the localized twist region. The flux tube is then defined as the expression given in equation (\ref{tube_field})  for $x_1^2+x_2^2\le a^2$ and zero for  $x_1^2+x_2^2> a^2$.  For the choice $b_0=5$, $w=0.9$, $l=0.1$ and $t_w=0.1$, the magnetic flux tube is visualized in Figure \ref{fig_tube}. Blue field lines illustrate twisted field lines near the boundary of the tube. Red field lines illustrate the localized region of strong twist at the centre of the tube at the apex. 
 \begin{figure}[h]
 \centering
{\includegraphics[width=9cm]{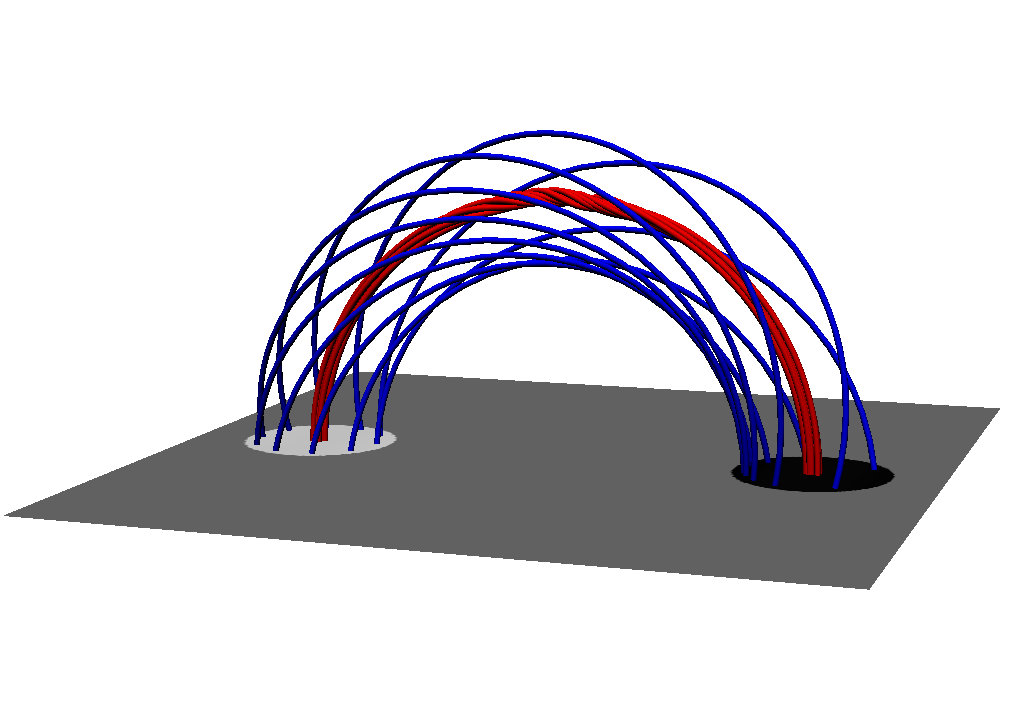}}
\caption{A toroidal shaped magnetic flux tube. Blue field lines are traced at the edge of the tube. Red field lines are traced a the centre of the tube and show a localized region of strong twist near the apex of the tube. The plane is at `height' $z_0$, on which white indicates positive magnetic field and black, negative. }\label{fig_tube}
\end{figure}

 \subsubsection{Helicity and winding inputs}
 To mimic emergence and submergence, we push the magnetic field defined by equation (\ref{tube_field}) through a horizontal boundary. For the emergence phase, we push the (rigid) magnetic flux tube through the photospheric plane until this plane is coincident with $z_0$ (as illustrated in Figure \ref{fig_tube}). For the submergence phase, the magnetic field is pulled back down through the photosphere at the same rate as it emerged. The submergence phase stops before the strong twist region, illustrated in Figure \ref{fig_tube} by the red field lines, submerges beneath the photosphere.
 
 We calculate $H_R$ and $L_R$ varying in time and normalize these quantities by their maxima, for ease of comparison. In determining these quadratures we impose a cut-off of $|B_z| = 0.01$ - any value beneath this cut-off is ignored to avoid numerical errors in the winding calculation \cite{prior2019interpreting, mactaggart2020helicity}. The results are displayed in Figure \ref{time_int}.
 
  \begin{figure}[h]
 \centering
{\includegraphics[width=9cm]{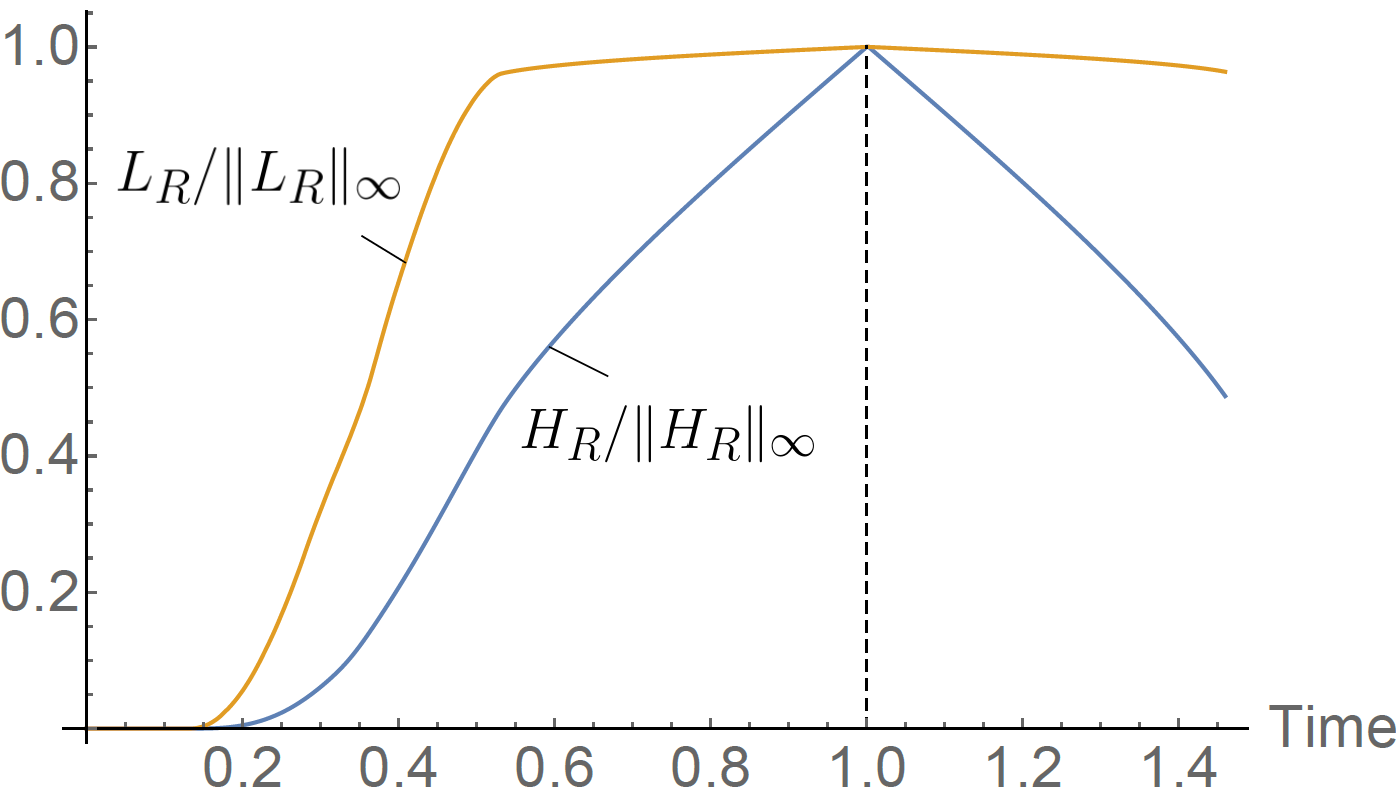}}
\caption{The helicity and winding time-integrated fluxes. Each quantity is weighted by its maximum value for presentation purposes. The vertical line separates the emergence and submergence phases.}\label{time_int}
\end{figure}
Focussing on the winding input in the emergence phase first, there is sharp increase followed by a levelling out after $t\approx 0.5$. The sharp rise is dominated by the detection of the localized region of strong twist. After this region has passed through the photosphere, the remaining twist in the emerging magnetic field is much weaker and its effect on the winding is thus much weaker. The helicity input in the emergence phase also shows an increase in time, with the gradient reaching a maximum just before $t\approx 0.5$ and then taking a smaller and almost constant value. The part of the curve with the highest gradient corresponds to the localized region of strong twist passing though the photosphere. The rest of the curve, with lower gradients, corresponds to the more weakly twisting field passing through the photosphere. Since helicity is weighted by magnetic flux, the identification of the input of topologically complex magnetic field above the photosphere is not as simple to detect as it is for the winding input. 

In the submergence phase, where part of the flux tube is pushed back down through the photosphere, the behaviour of the helicity and the winding is very different. The part of the flux tube that is pushed back down through the photosphere corresponds to the more weakly twisted field (the blue field lines in Figure \ref{fig_tube}). Since the helicity is weighted by magnetic flux, even if the topological complexity of the field lines is weak in this region, the helicity can still be large. Therefore, a large reduction is seen in the helicity due to submergence. For the winding, however, we see very little change due to submergence. This is because the winding is not biased by magnetic flux and so if only weak field line topology submerges, only a marginal decrease in the winding is recorded.  

This example has been set up to show that magnetic winding can identify particular regions of topological complexity which the helicity cannot. Thus a combination of winding and helicity can provide a more complete picture of the structure of emerging magnetic fields and, importantly, both quantities can be calculated in observations. {More realistic emergence studies in  \cite{prior2019interpreting, mactaggart2020helicity} further demonstrate the potential of the winding input rate time series as a metric for immediate event detection, the event being the emergence/submergence of sub-regions of highly twisted field emerging into the corona.}

\subsection{Influence of a moving boundary} 
So far, we have considered winding and helicity fluxes through a stationary and flat boundary $P$ that represents the solar photosphere. It is likely, however, that in solar observations, the magnetic field components recorded in planar magnetograms are actually at different heights. That is, the photosphere is likely to be a moving and non-uniform surface. In that case, what are the consequences for the  formulae that we have presented? For the magnetic winding flux, we will show that it is \emph{independent} of the geometry of a moving boundary and so the calculations using a flat plane can still be used in this situation. 

For the application of flux emergence through the photosphere, we will consider a photospheric boundary whose projection onto a horizontal plane is one to one,  as in Figure \ref{fig_moving}. We can consider more complicated boundary surfaces that fold back on themselves by defining coordinate systems on them and mapping them to flat planes \emph{via} conformal (angle-preserving) mappings. This process, however, is not necessary for the applications we have in mind and the following (simpler) description will be suitable for our purposes.  
\begin{figure}[h]
  \centering
\begin{tikzpicture}
\draw[thick]   (-4,0) node[above] {$P'$} -- (0.01,0) ;
\draw[thick] (4-0.01,0) -- (8,0); 
\draw[thick] plot[domain=0:4*pi, samples=80]  (\x/pi,{sin(\x r)}); 
%\draw[densely dashed]   (0,0) -- + (0,1.5)      (2,0) -- + (0,1.5);
%\draw[<->]      (0,1.3) -- node[above] {$\lambda$} + (2,0);
\draw[thick] (-4,-3) node[above] {$P$} -- (8,-3) ;
 \draw[->,thick] (6,-3) -- (6,-2.5)  node[right] {$\boldsymbol{e}_z$};
 \draw[->,thick] (6,0) -- (6,0.5)  node[right] {$\boldsymbol{e}_z$};
 \draw[->,thick] (2.82,0.5) -- (3.22,0.65)  node[right] {$\boldsymbol{N}$};
  \fill (0.5,1)  circle[radius=2pt];
  \fill (2.2,0.6)  circle[radius=2pt];
  \draw[thick,dashed]   (0.5,1) -- (2.2,0.6);
  \draw[thick,dashed]   (0.5,1) -- (0.5,-3);
  \draw[thick,dashed]   (2.2,0.6) -- (2.2,-3);
    \end{tikzpicture} 
    \caption{A representation of a non-uniform boundary $P'$ (which can also change in time) and the standard flat boundary $P$. Horizontal surfaces are orthogonal to $\ev_z$ and $\Nv$ represents the normal to $P'$, which is different at different locations. Two points on $P'$ are highlighted together with their projections onto $P$.}\label{fig_moving}
\end{figure}
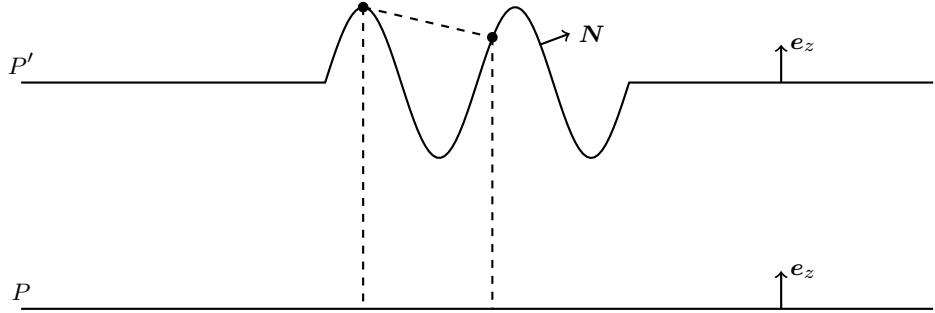

Let $P'$ represent a non-uniform simple surface that evolves in time. As shown in Figure \ref{fig_moving}, $P'$ is extended so that its boundary is horizontal. The purpose of it is to define an axis, the vector $\ev_z$ in this case, about which the winding of field lines is measured. Since winding is a two-dimensional (2D) measure, it requires a specific fixed normal vector. The vector $\Nv$ would not be suitable as it changes from point to point and $\Nv\equiv\ev_z$ only near the boundary of $P'$.  

Consider two points on $P'$, as shown in Figure \ref{fig_moving}. These two points (and any others on $P'$) are locations where field lines can intersect the surface. This could be an instantaneous moment of flux emergence, the movement of the boundary itself or a combination of both. The non-uniform boundary determines what part of the field line is used to calculate winding but its geometry does not enter explicitly into the winding rate calculation. This fact is simply because the relative angle of these points about a given origin is a purely 2D calculation, as is evident from 
\be
\frac{\d\theta}{\d t} = \frac{\d}{\d t}\arctan\left[\frac{(\xv-\yv)\cdot\ev_2}{(\xv-\yv)\cdot\ev_1}\right] = \ev_z\cdot\frac{(\xv-\yv)}{|\xv-\yv|^2}\times\left(\frac{\d\xv}{\d t}-\frac{\d\yv}{\d t}\right), 
\en
where $\xv$ and $\yv$ are horizontal vectors. Thus, to calculate the winding rate, every point on $P'$ can be projected orthogonally onto $P$ and the standard formula can be used, i.e.
\be
\frac{\d L_R}{\d t} = -\frac{1}{2\pi}\int_{P'}\int_{P'}\frac{\d}{\d t}\theta(\xv,\yv)\sigma(\xv)\sigma(\yv)\,\d^2x\,\d^2y = -\frac{1}{2\pi}\int_{P}\int_{P}\frac{\d}{\d t}\theta(\xv,\yv)\sigma(\xv)\sigma(\yv)\,\d^2x\,\d^2y.
\en
In this sense, the winding rate through $P'$ does not depend explicitly on the geometry of $P'$. The geometry of $P'$ selects what part of the field line is recorded but does not add any weighting to the winding rate calculation. 

{The input of magnetic helicity through the photosphere is dependent on the shape of $P'$. This is simply because helicity depends of the magnetic flux through $P'$ which depends on the shape of $P'$. Despite this dependency, however, calculations of helicity flux through different moving boundaries in flux emergence simulations \cite{mactaggart2020helicity} have demonstrated that the qualitative behaviour of helicity input in time is not strongly affected by the choice of boundary (assuming it is not placed unrealistically far from the photosphere region).}

\section{Summary}
In this work we have shown that magnetic winding is the topological underpinning of magnetic helicity for open magnetic fields. In essence, magnetic winding can be thought of as helicity without the magnetic flux weighting, thus representing a more fundamental topological quantity. Through examples of quasistatic resistive decay and flux emergence, we have demonstrated that, despite their close connection, helicity and winding can behave differently. Therefore, magnetic winding can provide different, and more detailed, information on field line topology than magnetic helicity. Due to the similarity, in terms of its mathematical expression, of winding to helicity, there are exciting opportunites to use winding practically, in both simulations and observations, to provide a deeper understanding of magnetic field topology.

\section*{Appendix}
% For completeness and for the convenience of the reader, we include how to calculate the toroidal flux tubes that we have used in this study. This approach was first used in (ref).

The quasistatic induction equation is defined by the following identity:
\begin{equation}
\nabla\times (\lambda \nvec + \alpha\Bv) = \nabla\lambda\times \nvec + \lambda \nabla\times \nvec  + \nabla\alpha \times \Bv + \alpha\nabla\times \Bv.
\end{equation}
The aim is to decompose this into components parallel to the field and those perpendicular, which must either be along the direction of $ \nvec$ or along the direction of the Lorentz force  $\hat{\Fv}_l$. To do so we use the following identities:
\begin{align}
&\alpha\nabla\times \Bv= \alpha\lambda \nvec+\alpha^2\Bv,\\
&(\nabla\lambda\times\Bv_{f\perp})\cdot  \nvec = 0,\\
&(\nabla\lambda\times\Bv_{f\perp})\cdot \Bv =\nabla\lambda\cdot(  \nvec\times \Bv) ={\vert\Bv\vert^2}\nabla\lambda\cdot \hat{\Fv}_l,\\
&(\nabla\lambda\times \nvec)\cdot \hat{\Fv}_l = -{\nabla\lambda\cdot\Bv}.
\end{align}
Finally we consider the vector $\nabla\times \nvec$. This is the rotation of the vector normal to both the Lorentz force and the magnetic field. We can express this vector in the local orthogonal basis $(\Bv,\hat{\Fv},\nvec)$ as follows:
\begin{equation}
\nabla\times \nvec = \omega_b \Bv+\omega_f\hat{\Fv}_{l} + \omega_{fl}\nvec,
\end{equation}
where 
\begin{equation}
\omega_b = \frac{\nabla\times  \nvec\cdot{\Bv}}{\vert \Bv\vert^2} 
\end{equation}
represents the rotation of the Lorentz force vector around the field line,
\begin{equation}
\omega_f = \nabla\times \nvec\cdot{\hat{\Fv}_l},
\end{equation}
represents the rotation of field $\Bv$ around the direction of the Lorentz force, and 
\begin{equation}
 \omega_{fl}= \frac{\nabla\times \nvec\cdot{\nvec}}{\vert \Bv\vert^2},
\end{equation}
represents the rotation of the pair $(\Bv,\hat{\Fv}_l)$ around the direction of the field  $ \nvec$.

\enlargethispage{20pt}

\ethics{Not applicable.}

\dataccess{This article has no additional data.}

\aucontribute{CP contributed to mathematical modelling, performed numerical calculations, plotted the figures and reviewed the manuscript; DM contributed to mathematical modelling, performed numerical calculations and drafted the manuscript. Both authors approved the final version and agree to be accountable for all aspects the work.}

\competing{No competing interests.}

\funding{Not applicable.}

\ack{We thank Anthony Yeates for helpful discussions.}

%%%%%%%%%% Insert bibliography here %%%%%%%%%%%%%%

\vskip2pc

%\noindent {\bf Please follow the coding for references as shown below.}
%
%\begin{thebibliography}{9}
%
%\bibitem{1} Allwood JM, Cullen JM. 2011 \textit{Sustainable materials:  with both eyes open}.
%Cambridge, UK: UIT Cambridge. See \href{http://www.withbotheyesopen.com}{http://www.withbotheyesopen.com}.
%
%\bibitem{2}  MacKay DJC. 2008  \textit{Sustainable energy:  without the hot air}.
% Cambridge, UK: UIT Cambridge. See \href{http://www.withouthotair.com}{http://www.withouthotair.com}.
%
%\bibitem{3} Gallman PG. 2011  \textit{Green alternatives and national energy strategy: the facts
% behind the headlines}.  Baltimore,\ MD: Johns Hopkins University Press.
%
%\bibitem{4} MacKay DJC. 2013.  Solar energy in the context of energy use, energy transportation, and
% energy storage. \textit{Proc. R. Soc. A} \textbf{371}.
%
%\end{thebibliography}

\bibliographystyle{RS} %%%% .BST file

\bibliography{windingBib} %%%%% .Bib file

\end{document}